\documentclass[submission,copyright,creativecommons]{eptcs}
\providecommand{\event}{ICE 2026} 

\usepackage{iftex}
\usepackage{amsmath}
\usepackage{amsthm}
\usepackage{amssymb}
\usepackage{mathpartir}
\usepackage{mathtools}
\usepackage{xcolor}
\usepackage{cleveref}

\newcommand{\co}[1]{\overline{#1}}
\newcommand{\nil}{\mathbf{0}}
\newcommand{\ckey}[1]{\textcolor{magenta}{[#1]}}
\newcommand{\tup}[1]{\langle#1\rangle}
\newcommand{\std}{\mathtt{std}}
\newcommand{\key}{\mathtt{key}}
\newcommand{\fw}[1]{\xrightarrow{#1}}
\newcommand{\bk}[1]{\xhookrightarrow{#1}}

\newcommand{\ctx}{\mathtt{C}}
\newcommand{\red}[1]{\xmapsto{#1}}
\newcommand{\causeq}{\asymp}
\newcommand{\po}[1]{\mathtt{po}(#1)}
\newcommand{\ind}{\,\mathtt{I}\,}

\newcommand{\rul}[2]{\textcolor{#1}{\underline{\textcolor{black}{#2}}}}
\newcommand{\lang}{\text{CCS}^\mathtt{a}}
\newcommand{\rlang}{\mathtt{r}\lang}

\newcommand{\added}[1]{#1}

\newtheorem{theorem}{Theorem}
\newtheorem{lemma}[theorem]{Lemma}
\newtheorem{proposition}[theorem]{Proposition}
\newtheorem{corollary}[theorem]{Corollary}
\newtheorem{property}[theorem]{Property}
\newtheorem{example}[theorem]{Example}

\theoremstyle{definition}
\newtheorem{definition}[theorem]{Definition}

\ifpdf
  \usepackage{underscore}         
  \usepackage[T1]{fontenc}        
\else
  \usepackage{breakurl}           
\fi

\usepackage{todonotes}

\title{On Asynchrony and Reversibility in CCS \thanks{This work has been supported by 
    the   European Union through the MSCA SE project QCOMICAL (Grant Agreement ID: 101182520).
    }}
    \author{ Hern\'an Melgratti
\institute{ICC - Universidad de Buenos Aires - Conicet, Argentina}
\and
Claudio Antares Mezzina
\institute{Dipartimento di Informatica, \\Universit\`a di Bari Aldo Moro, Italy}
\and
G. Michele Pinna
\institute{Dipartimento di Matematica e Informatica, \\Universit\`a di Cagliari, Italy}
}
\def\titlerunning{On Asynchrony and Reversibility in CCS}
\def\authorrunning{H. Melgratti, C. A. Mezzina, G. M. Pinna}
\begin{document}
\maketitle

\begin{abstract}
Asynchronous communication is a fundamental feature of modern distributed systems, where messages are emitted without requiring immediate synchronization with receivers. In process calculi, this behaviour is typically modelled by separating message emission from message consumption. At the same time, reversible computation has emerged as an important paradigm for analysing concurrent systems, enabling computations to be undone while preserving causal dependencies between actions. While reversible semantics have been extensively studied for synchronous process calculi such as CCS, their integration with asynchronous communication remains largely unexplored.

In this paper we investigate the interaction between asynchrony and reversibility in the setting of CCS. We first introduce $\lang$, an asynchronous variant of CCS in which output actions generate explicit message entities that can later be consumed by matching input actions. We then define $\rlang$, a reversible extension of $\lang$ obtained by adapting the framework of Phillips and Ulidowski \cite{ccsk}. In $\rlang$, prefixes and messages are annotated with unique keys that record message emission and consumption events, allowing computations to be reversed while preserving causal dependencies.
We show that the resulting reversible semantics satisfies causal consistency, ensuring that computations can be reversed exactly up to causal equivalence. The proof relies on the axiomatic framework for reversible computation proposed by Lanese et al. \cite{LanesePU24}.
\end{abstract}

\section{Introduction}

Concurrency theory provides the foundations for understanding systems composed
of multiple interacting computational entities. Over time, several formal
approaches have been developed, including Petri nets, event structures, and
process calculi and algebras. Among these, the Calculus of Communicating
Systems (CCS) \cite{Milner80} is a seminal reference for reasoning about
communication, synchronisation, and behavioral equivalence in distributed
systems. CCS models systems as collections of sequential processes that
interact by exchanging messages through communication channels.
%
Concretely, a process ready to receive a message on channel $a$ and then
continue as $P$ is written $a.P$, while a process ready to send a message on
$a$ and then continue as $Q$ is written $\overline{a}.Q$. Here, $a$ denotes an
input prefix and $\overline{a}$ an output prefix; prefixes capture the idea
that a process must perform the corresponding action before evolving into its
continuation. Processes can be composed in parallel using the operator
$\parallel$, so that $P \parallel Q$ represents the concurrent execution of
$P$ and $Q$. When complementary actions are enabled, they may interact via a
reduction step, for instance:
$a.P \parallel \co{a}.Q \fw{\tau} P \parallel Q$, which models the
synchronisation and consumption of the communication action.
This reduction embodies a synchronous view of communication, where message
exchange occurs atomically through rendezvous.

In contrast, real-world distributed systems typically communicate
asynchronously, with messages sent without requiring immediate synchronisation
with a receiver. Buffering mechanisms further enable communication to occur
even when the receiver is not ready. Modern distributed programming languages
embrace this paradigm. For instance, in Go (\cite{Gocit}), communication
through channels allows processes to send messages that may be buffered and
consumed later by receivers. Similarly, Erlang (\cite{Erlangcit}) relies on
asynchronous message passing, where messages are placed in process mailboxes
and consumed when the receiver performs a matching input. These mechanisms
naturally correspond to a semantics in which message emission and message
consumption are distinct events.

This mismatch between the synchronous semantics of CCS and the asynchronous
nature of practical systems has motivated  the development of asynchronous process
calculi. There exist two main approaches to modelling asynchrony, differing in
how (non-blocking) output actions are represented \cite{BorealeNP98}. In the
first approach, outputs are treated as processes, typically enforced via a
syntactic restriction whereby output prefixes have no continuation. This is
the case, for example, in asynchronous CCS (ACCS
\cite{BorealeNP98,BorealeNP02}) and the asynchronous $\pi$-calculus
\cite{api}. In the second approach, output prefixes are allowed to have
continuations, and the operational semantics generates the corresponding
message processes. This is the case in CCS with Linda-like communication
primitives \cite{BusiGZ98} and in asynchronous session types \cite{HondaYC08}.

In this work, we follow the second approach: we retain the standard CCS prefix
syntax for outputs and let the operational semantics generate explicit message
processes representing messages in transit. This choice keeps the language
close to the structure of programming languages while modelling asynchronous
communication through the presence of explicit message entities. In
particular, output actions are no longer interpreted as one half of a
rendezvous, but as independent emission events that generate persistent
message entities. Operationally, this corresponds to an emission rule where
the output action materialises as a message process $\tup{a}$ that can later
be consumed by a matching input.
Communication is thus decomposed into two causally distinct steps: message
emission and message reception. In this model, the communication medium can be
understood as a ``bag'' of messages (output actions), waiting to be consumed
by corresponding input actions. Communication is then captured by the
following two reductions, representing output on the left and input on the
right:
\begin{eqnarray}
  \co{a}.P \fw{\tau} P \parallel \tup{a} \label{rule:emit}
  \hspace{3cm}
   a.P \parallel \tup{a}  \fw{\tau} P
\end{eqnarray}
%

%

The decoupling of message emission and consumption makes it harder to analyse
executions, diagnose errors, and recover from undesired states. This motivates
the study of mechanisms that allow computations not only to proceed forward,
but also to be reversed in a controlled manner.
Reversibility has emerged as a fundamental concept in the theory of
computation~\cite{cost}. \added{Reversibility allows a system to undo
its computational steps. In a sequential setting this is easily achievable by using the so-called
Bennett embedding \cite{bennett}. In concurrent systems, however, reversibility is more challenging to define and implement, due to the presence of independent and potentially interacting processes. In such settings, it is necessary to keep track of causal dependencies between actions in order to ensure that only those steps that do not violate causality can be reversed. 
}
This notion, known as \emph{causal-consistent
  reversibility}~\cite{rccs}, has proved useful in debugging concurrent
systems~\cite{fase,HoeyU19,LaneseSUS22}, where reversing erroneous
computations is more natural than replaying from scratch; in biochemical
modelling~\cite{MelgrattiMP22,KuhnU22}, given that many biochemical reactions
are inherently reversible; in fault tolerance, where recovery requires
controlled rollback to consistent
states~\cite{LaneseMV26,MezzinaTY25,VassorS18,LaneseLMSS13}; in quantum computing, where
reversibility is forced by the laws of physics \cite{FrankS21}; and in low-energy computation
inspired by thermodynamic considerations~\cite{Landauer61}.
In the last decade, several reversible extensions of process calculi have been
proposed, enriching classical models with mechanisms to record and undo
computational history while preserving causal dependencies between actions. A
seminal contribution in this direction is Reversible CCS (RCCS) by Danos and
Krivine~ \cite{rccs}, which extends CCS with explicit memories that record
past actions and enable undos. Phillips and Ulidowski proposed a general
framework for equipping process calculi with reversibility, applicable to
calculi whose operational semantics is specified in a suitable SOS \added{(structural operational semantics)} 
format~\cite{ccsk}. The main idea of their approach is to keep the structure
of processes unchanged during execution--making the operators of the calculus
\emph{static}--and to record the information needed for reversing computations
by attaching unique \emph{keys} to communications. When this framework is
instantiated for CCS, it yields CCSK (CCS with communication keys).
This model is closely related to Reversible CCS (RCCS): the two calculi are
equivalent in terms of labelled transition system (LTS)
isomorphism~\cite{acta}, differing only in the way history information is
recorded.
Despite the significant progress in reversible process
calculi~\cite{MezzinaSGHHCLMSSU20}, the integration of reversibility with
asynchronous communication remains less explored, with the exception of
asynchronous RCCS \cite{CardelliL11}, and asynchronous higher order $\pi$-calculus \cite{LaneseMS16}. \added{ Both \cite{CardelliL11,LaneseMS16} fall in the first approach of modelling asynchrony.}

 In asynchronous
settings, where communication is split into distinct emission and consumption
events, this decoupling introduces additional challenges for reversible
semantics. In particular, reversing a computation requires restoring not only
the processes involved, but also the correct configuration of messages in
transit, while preserving the causal dependencies between emission and
reception events.
As a result, designing a reversible semantics for asynchronous communication
demands a careful treatment of the information used to track causality.
Existing approaches developed for synchronous calculi do not directly apply in
this setting. For instance, the framework of Phillips and Ulidowski
\cite{ccsk} relies on a fixed SOS format, which is not compatible with rules
such as emission, where the syntactic structure of processes may change. On
the other hand, RCCS \cite{rccs} is based on memories that record complete
synchronisation events, whereas asynchronous communication involves partial
events that are split across emission and consumption.

In this paper, we investigate the integration of reversibility with asynchronous communication in the setting of CCS. We first introduce $\lang$,  an asynchronous variant of CCS in which output actions generate explicit message processes that can be consumed at a later stage. We then define a reversible extension of $\lang$ enriched with key annotations that record message emission and consumption events.

The remainder of the paper is structured as follows. In  Section \ref{sec:accs} the syntax and the semantics of $\lang$ are presented; while in Section \ref{sec:raccs} the syntax and semantics of
$\rlang$ are presented. In Section \ref{sec:prp}, we discuss the main properties of $\rlang$, while Section \ref{sec:ex} demonstrates how the framework can be used to model the leader election problem. Finally,
Section \ref{sec:conc} 
concludes the paper with a discussion of the related work.

\section{Asynchronous CCS}\label{sec:accs}

In this section, we introduce an asynchronous variant of CCS, which we refer
to as $\lang$. In $\lang$, message sending and reception are not
synchronised~\cite{BaldanBGM15}. Messages are transmitted through a medium
until they reach their destination. Consequently, sending is non-blocking,
allowing a process to send regardless of the receiver's state, whereas
receiving is blocking because processes must wait for a message to arrive.

\begin{figure}[th]
\begin{align*}
(\mathit{Prefixes}) \quad \mu,\eta ::= \,\,& a \,\mid \, \overline{a} \,\mid \, \tau\\
(\mathit{Processes}) \quad P,Q ::=\,\, &\nil \,\mid \,\displaystyle{\sum}_{i\in I}\mu_i.P_i \,\mid \, P \parallel Q  \,\mid \, (\nu a) P \,\mid \, A\\
(\mathit{Systems})\quad R,S ::= \,\,&  \tup{a} \,\mid \, P \,\mid \, R \parallel S \,\mid \,(\nu a) R
\end{align*}
\caption{$\lang$ syntax}
\label{fig:accs-syn}
\end{figure}

Let $\mathcal{N} = \{a, b, c, \ldots\}$ be a set of \emph{names}, and let
$\overline{\mathcal{N}} = \{\overline{a} \mid a \in \mathcal{N}\}$ be the set of its corresponding
\emph{co-names}.
Elements in $\mathcal{N}$ represent input actions, while elements in $\overline{\mathcal{N}}$ represent output actions. We use $\alpha, \beta, \ldots$ to range over $\mathcal{N} \cup \overline{\mathcal{N}}$, and we assume $\overline{\overline{\alpha}} = \alpha$, calling
$\alpha$ and $\overline{\alpha}$ \emph{complementary} actions.
We also consider a special \emph{silent action} $\tau$, \added{and assume that $\tau \not \in \mathcal{N} $}.
The syntax of $\lang$ is reported in Figure \ref{fig:accs-syn}.
A prefix (or action), denoted by  $\mu, \eta$,  is an element of  $\mathtt{Act} = \mathcal{N} \cup \overline{\mathcal{N}} \cup \{\tau\}$.
We let $P, Q, \ldots$ range over \emph{Processes}.
The term $\nil$ represents the idle process, while
 $\sum_{i\in I}\mu_i.P_i$ represents a non-deterministic
choice that begins by performing some prefix $\mu_i$ and continues as $P_i$.
We write $\mu_i.P_i$ for a singleton sum where $I = \{i\}$, and use
$\mu_1.P_1 + \sum_{i \in I \setminus \{1\}} \mu_i.P_i$ to distinguish a specific
branch, assuming guarded choice is associative and commutative.
The term $P \parallel Q$ represents the parallel composition of processes $P$ and $Q$,
while $(\nu a)P$ denotes the process $P$ where the name $a$ is restricted.
We use $A$ for process constants, assuming a unique definition $A \triangleq P$
for each $A$. These constants are employed to model recursive behaviours.

\emph{Systems}, ranged over by $R, S$, consist of the parallel composition of
sent messages (each denoted by $\tup{\cdot}$) and processes.
Analogously to processes, we denote by $R \parallel S$ the parallel composition
of systems, and by $(\nu a)R$ the system in which the name $a$
is restricted.

We use  \emph{term} to refer to either a system or a process.

\begin{definition}[$\lang$ LTS]
The operational semantics of $\lang$ is defined
as the LTS $(Systems,\mathtt{Act},\fw{})$ where the transition relation $\fw{}$ is the smallest relation induced by the rules reported in Figure \ref{fig:accs-sem}.
\end{definition}

\begin{figure}[ht]
\begin{mathpar}
  \inferrule*[left=\scriptsize{(Out)}]
  {}{ \co{a}.P \fw{\tau} P \parallel \tup{a} } \and
  \inferrule*[left=\scriptsize{(Msg)}]
  {}{ \tup{a} \fw{\co a}  \nil } \and
  \inferrule*[left=\scriptsize{(In)}]
  {}{ a.P \fw{ a} P } \and
    \inferrule*[left=\scriptsize{(Tau)}]
  {}{ \tau.P \fw{ \tau} P } \and
\and
      \inferrule*[left=\scriptsize{(Sum)}]
  {j\in I \and P_j \fw{\mu} P'  }
  { \displaystyle{\sum}_{i\in I}P_i \fw{ \mu} P' }
  \and
          \inferrule*[left=\scriptsize{(Par)}]
  {R \fw{\mu} R' }{ R\parallel S \fw{ \mu} R'\parallel S }
\and
      \inferrule*[left=\scriptsize{(Syn)}]
  {R \fw{\alpha} R'  \and S \fw{\co{\alpha}}S'}{ R\parallel S \fw{ \tau} R'\parallel S' }
  \and
        \inferrule*[left=\scriptsize{(Res)}]
  {R \fw{\mu} R' \and \mu \not \in \{a,\co{a}\}}{ (\nu a)R \fw{ \mu} (\nu a)R'} \and
  \inferrule*[left=\scriptsize{(Const)}]
    {A \triangleq P \and P \fw{\mu} P'}{ A \fw{ \mu} P'}
\end{mathpar}
\caption{$\lang$ semantics, symmetric rules for $\textsc{Par}$ and $\textsc{Syn}$ are omitted.}
\label{fig:accs-sem}
\end{figure}

As usual, we write $S \xrightarrow{\mu} S'$ to
denote a transition in $\xrightarrow{}$. The operational rules follow those of
CCS~\cite{Milner80}, with the exception of (\textsc{Out}) and (\textsc{Msg}),
which handle asynchrony. Under rule (\textsc{Out}), the output-prefixed term
$\co a.P$ emits a message $\tup{a}$ which is placed in the communication
media, which is represented as the parallel composition of the message with
the continuation $P$. Note that this emission results in a silent action. Rule
(\textsc{Msg}) models the consumption of a message from the communication
media by a input-prefixed process $a.P$.

\begin{example}\label{ex:intro}
 Let $R = (a.\nil + b.\nil) \parallel \co{a}.\co{b}.\nil $. A possible execution is as follows
 \begin{align*}
 R\fw{\tau} \;& (a.\nil + b.\nil)    \parallel \co{b}.\nil \parallel \tup{a} \\
 \fw{\tau} \;&
 (a.\nil + b.\nil)  \parallel \tup{a} \parallel \tup{b}\parallel \nil\\
 \fw{\tau}  \; &\nil \parallel \nil \parallel \tup{b}
 \end{align*}
 where the second process first emits messages $a$ and $b$ in succession,
 followed by the first process consuming $a$.
\end{example}

\section{Reversible asynchronous CCS}\label{sec:raccs}

In this section, we introduce $\rlang$, a reversible variant of $\lang$. 
\added{To make $\lang$ reversible, we exploit the ideas of \cite{ccsk}: decorate all the executed actions
with keys, and make all the operators of the calculus \emph{static}. Operators such as the prefixing ` $.$' and the $\sum$ are dynamic, and hence forgetful as they disappear after the reduction.
Hence, we will instrument the semantics of $\rlang$ to make all the operators static.
}

The
syntax of $\rlang$ is presented in Figure~\ref{fig:raccs-syn}. Unlike $\lang$,
messages and prefixes in $\rlang$ may be decorated with unique \emph{keys}.
We let $\mathcal{K} = \{i, j, \ldots\}$ be the set of keys.
A message $\ckey{i}\tup{a}$ represents an output produced by a prefix marked
with $i$, while $\ckey{i}\tup{a}\ckey{j}$ denotes a message (originally
produced by $i$) that has been consumed by an input prefix decorated with $j$.
Notably, consumed messages \emph{remain} in the system as inactive processes;
they act as memories, storing the computational history required to revert
both the consumption and the original emission. \emph{Marked prefixes},
denoted by $\mu\ckey{i}$, indicate that the action $\mu$ has already been
executed. 
Note that, since all the operators are static, also sums are static.
This means that a choice may be composed with some marked prefix and some unmarked prefixes.
This will become clear once we introduced the operational semantics.

\begin{figure}[th]
  \begin{align*}
    (Prefixes) \quad \mu ::= & a \,\mid \, \overline{a} \,\mid \, \tau\\
    (Processes) \quad P,Q ::= &\nil \,\mid \,\displaystyle{\sum}_{i\in I}\pi_i.P_i \,\mid \, P \parallel Q  \,\mid \, (\nu a) P \,\mid \, A\\
    (Systems)\quad R,S ::= \,\,&  \ckey{i}\tup{a} \,\mid \,  \ckey{i}\tup{a}\ckey{j} \,\mid \,P \,\mid \, R \parallel S \,\mid \,(\nu a) R\\
    (Actions) \quad \pi ::= & \mu  \,\mid \, \mu\ckey{i}
  \end{align*}
  \caption{$\rlang$ syntax}
  \label{fig:raccs-syn}
\end{figure}

We remark that $\rlang$ generalizes $\lang$, in the sense that each $\lang$
term is also an $\rlang$ term; specifically, these are the $\rlang$ terms that
do not contain keys. We call such processes \emph{standard}, as formally
defined below.

\begin{definition}[System and process keys]
  The set of keys of a $\rlang$ term is inductively defined as follows:
  \begin{align*}
    &\key(\ckey{i}\tup{a}) = \{i\} & \key(\ckey{i}\tup{a}\ckey{j}) =  \{i,j\}
    && \key(R \parallel S) = \key(R) \cup \key(S) \\
    &\key(R\backslash a) = \key(R) &\key(\sum_{i \in I} P_i) = \bigcup_{i\in I} \key(P_i)
    && \key(P \parallel Q) = \key(P) \cup \key(Q) \\
    & \key(\pi\ckey{i}.P) = \{i\} \cup \key(P) &\key(P\backslash a) =  \key(P) 
    &&\key(A) = \emptyset
  \end{align*}
\end{definition}

\begin{definition}[Standard System]\label{def:std}
  A system $R$ is said to be standard, written $\std(R)$, if it contains no
  key, i.e., $\key(R) = \emptyset$.
\end{definition}

The operational semantics of $\rlang$ is defined by two transition relations:
$\fw{}$, which captures forward computation, and $\bk{}$, which captures
backward computation.

\begin{definition}\label{de:ltlrccsa}[$\rlang$ semantics]
  The operational semantics of $\rlang$ is formally defined as the LTS
  $(\mathit{Systems}, \mathtt{Act} \times \mathcal{K}, \red{\ })$, where
  ${\red {\ }} = {\fw{\ } \cup \bk{\ }}$, and $\fw{\ }$ and $\bk{\ }$ are the
  smallest relations generated by the rules reported in
  Figures~\ref{fig:semfw} and \ref{fig:sembk}, respectively.
\end{definition}

We begin by commenting the forward rules shown in Figure~\ref{fig:semfw}. All
rules, except (\textsc{Act}) \added{and (\textsc{Sum})}, mirror those of the non-reversible language in
Figure~\ref{fig:accs-sem}, but are enriched with key annotations to track
causality. Moreover, executed prefixes are persistent, i.e., they are not
consumed by reduction.
\begin{figure}[th]
\begin{mathpar}
  \inferrule*[left=\scriptsize{(Out)}]
  {\std(P)}{ \co{a}.P \fw{\tau[i]} \co{a}\ckey{i}.P \parallel \ckey{i}\tup{a} }
  \and
  \inferrule*[left=\scriptsize{(Msg)}]
  {}{ \ckey{i}\tup{a} \fw{\co a[j]}  \ckey{i}\tup{a}\ckey{j} }
  \and
  \inferrule*[left=\scriptsize{(In)}]
  {\std(P)}{ a.P \fw{ a[i]} a\ckey{i}.P }
  \and
    \inferrule*[left=\scriptsize{(Tau)}]
  {\std(P)}{ \tau.P \fw{ \tau[i]} \tau\ckey{i}.P }
  \and
    \inferrule*[left=\scriptsize{(Act)}]
  {P \fw{ \eta[j]} P' \and i\neq j}
  { \mu\ckey{i}.P \fw{ \eta[j]} \mu\ckey{i}.P' }
  \and
\added{
      \inferrule*[left=\scriptsize{(Sum)}]
  { P_z \fw{\mu[k]} P_z'  \and \forall j \in (I\setminus \{z\}) .(P'_j = P_j \wedge \std(P_j))}
  { \displaystyle{\sum}_{i\in I}P_i \fw{ \mu[k]} \displaystyle{\sum}_{i\in I}P'_i }
  }
  \and
          \inferrule*[left=\scriptsize{(Par)}]
  {P \fw{\mu[i]} P' \and i \not\in \key(Q)}{ P\parallel Q \fw{ \mu[i]} P'\parallel Q }
\and
      \inferrule*[left=\scriptsize{(Syn)}]
  {P \fw{\alpha[i]} P'  \and Q \fw{\co{\alpha}[i]}Q'}{ P\parallel Q \fw{ \tau[i]} P'\parallel Q' }
  \and
        \inferrule*[left=\scriptsize{(Res)}]
  {P \fw{\mu[i]} P' \and \mu \not \in \{a,\co{a}\}}{ (\nu a)P \fw{ \mu[i]} (\nu a)P'} \and
  \inferrule*[left=\scriptsize{(Const)}]
    {A \triangleq P \and P \fw{\mu[i]} P'}{ A \fw{ \mu[i]} P'}

\end{mathpar}
\caption{$\rlang$ forward semantics, symmetric rules for \textsc{Par} and \textsc{Syn} are omitted.}
\label{fig:semfw}
\end{figure}

Rule (\textsc{Out}), which models message emission, uses a key $i$ to bind the
emitted message, which is tagged with $i$, to the prefix that produced it.
Notably, the prefix is preserved by the reduction and annotated with $i$,
indicating that it has been executed. This establishes a link between the
message and its generating prefix, which is essential for reversing this
computation step. Note that the used key is reflected in the transition label.
Rule (\textsc{Msg}) governs message consumption. In contrast to synchronous
reversible semantics, where causality is associated with a single
communication event, here causality is distributed across two distinct steps:
emission and consumption. To account for this, it is necessary to record both
the sender and the receiver of a message. Accordingly, a consumed message is
represented as $\ckey{i}\tup{a}\ckey{j}$, indicating that it was emitted by an
output prefix marked with $i$ and consumed by an input prefix marked with $j$.
Rule (\textsc{In}) uses a key for decorating the input prefix. Finally, rule
(\textsc{Tau}) records the execution of a $\tau$-prefix.

Rule (\textsc{Act}) is the only additional rule, introduced to account for the
persistency of executed prefixes. It inductively enables actions to be
propagated through prefixes that have already been executed.
The second premise requires that the key $i$ decorating the already executed
prefix be different from the key $j$ associated with the action that is being
performed. This condition is essential to ensure that each reduction step uses
a fresh key, thereby preserving the uniqueness of keys.
Rule (\textsc{Sum}) captures the behaviour of a non-deterministic choice. If
the process is standard (see Definition~\ref{def:std}), then the choice has
not yet been resolved \added{and all the branches do not contain marked prefixes}. \added{In this case,  any of the branches may proceed.} \added{Once a branch executes, its prefix is marked while the other branches will remain as a decoration of the executed branch. From now on, this branch is allowed to execute forward (and backward), and the other branches remain as a decoration of the executing branch.}
 If the choice is resolved, only the selected (non-standard) branch is allowed to
execute.
In this way, the choice operator, together with the discarded branches, is
preserved in the syntax and acts as a decoration of the forward computation.
This information is crucial for enabling the reversal of a choice\added{, and this is how the sum operator is rendered \emph{static}.}
The additional premise in rule (\textsc{Par}) ensures that the key in the
label, that is, the one assigned to the action being executed, has not been
previously used within the parallel component $Q$. This condition enforces the
uniqueness of keys.
Rules (\textsc{Syn}), (\textsc{Res}) and (\textsc{Const}) are 
straightforward adaptations of the corresponding rules in Figure
\ref{fig:accs-sem}.


\begin{example}\label{ex:fw}
  Consider the system $R$ and the execution shown in Example~\ref{ex:intro}.
  The same execution can be mimicked in $\rlang$ as follows
  \begin{align*}
    R\fw{\tau[i]} \;
    & (a.\nil + b.\nil) \parallel \co{a}\ckey{i}.\co{b}.\nil \parallel \ckey{i}\tup{a} \\
    \fw{\tau[j]} \;
    & (a.\nil + b.\nil) \parallel \co{a}\ckey{i}.(\co{b}\ckey{j}.\nil \parallel
      \ckey{j}\tup{b})\parallel \ckey{i}\tup{a}\\
    \fw{\tau[w]} \;
    & (\rul{red}{a\ckey{w}}.\nil + b.\nil) \parallel
      \rul{red}{\co{a}\ckey{i}}.(\co{b}\ckey{j}.\nil \parallel \ckey{j}\tup{b})\parallel
      \rul{red}{\ckey{i}\tup{a}\ckey{w}} = S
  \end{align*}

  where the parts of the system involved in the communication are
  \rul{red}{underlined}. As indicated by the message
  $\ckey{i}\tup{a}\ckey{w}$, the communication takes place in two distinct
  steps: the emission, decorated by $i$, and the consumption, decorated by $w$

  It is worth noting that prefixes are not consumed by reduction, but rather
  marked with keys. Moreover, emitted messages are placed in parallel with the
  emitting prefix. This changes the structure of the process and is therefore not
  compliant with the framework of~\cite{ccsk}, which requires operators to
  remain static.
  Similarly, discarded branches of a choice are not removed; instead, they are
  retained in the syntax as a form of decoration. For instance, after the
  consumption of the message $\ckey{i}\tup{a}$, the process $a.\nil + b.\nil$
  evolves into $a\ckey{w}.\nil + b.\nil$, where the branch $+, b.\nil$ serves
  as a decoration/context of the process $a\ckey{w}.\nil$.
\end{example}

The backward rules are reported in Figure~\ref{fig:sembk}. Each forward
rule has a corresponding backward rule that restores the configuration
preceding the action. In particular, reversing an emission removes the
generated message and restores the original output prefix, while reversing a
message consumption transforms a consumed message $\ckey{i}\tup{a}\ckey{j}$
back into an available message $\ckey{i}\tup{a}$. The symmetry between forward
and backward rules ensures that the system satisfies the Loop Lemma (Lemma
\ref{lm:loop}), stating that every forward transition can be undone.
The premise $\std(P)$ in rules $(\textsc{Out}^{\bullet})$,
$(\textsc{In}^{\bullet})$, and $(\textsc{Tau}^{\bullet})$ ensures that a
sequential process having several previously executed prefixes reverses the
last one first. \added{Let us note that  the premises in rule (\textsc{Par}$^\bullet$)
check whether the key is not shared with another component. This check is instrumental
to force a synchronisation back, e.g. in the case the key appears also in the component in parallel one has to use rule (\textsc{Syn}$^\bullet$).
}

\begin{figure}[th]
\begin{mathpar}
  \inferrule*[left=\scriptsize{(Out$^\bullet$)}]
  {\std(P)}{  \co{a}\ckey{i}.P \parallel \ckey{i}\tup{a} \bk{\tau[i]}  \co{a}.P}
  \and
  \inferrule*[left=\scriptsize{(Msg$^\bullet$)}]
  {}{ \ckey{i}\tup{a}\ckey{j}  \bk{\co a[j]}\ckey{i}\tup{a}   }
  \and
  \inferrule*[left=\scriptsize{(In$^\bullet$)}]
  {\std(P)}{ a\ckey{i}.P \bk{ a[i]} a.P }
  \and
    \inferrule*[left=\scriptsize{(Tau$^\bullet$)}]
  {\std(P)}{ \tau\ckey{i}.P \bk{ \tau[i]} \tau.P }
  \and
      \inferrule*[left=\scriptsize{(Act$^\bullet$)}]
  {P \bk{ \eta[j]} P' \and i\neq j}
  { \mu\ckey{i}.P \bk{ \eta[j]} \mu\ckey{i}.P' }
   \and
   
   \added{
      \inferrule*[left=\scriptsize{(Sum$^\bullet$)}]
  { P_z \bk{\mu[k]} P_z'  \and \forall j \in (I\setminus \{z\}) .(P'_j = P_j \wedge \std(P_j))}
  { \displaystyle{\sum}_{i\in I}P_i \bk{ \mu[k]} \displaystyle{\sum}_{i\in I}P'_i }
  }
  \and
          \inferrule*[left=\scriptsize{(Par$^\bullet$)}]
  {P \bk{\mu[i]} P' \and i \not\in \key(Q)}{ P\parallel Q \bk{ \mu[i]} P'\parallel Q }
\and
      \inferrule*[left=\scriptsize{(Syn$^\bullet$)}]
  {P \bk{\alpha[i]} P'  \and Q \bk{\co{\alpha}[i]}Q'}{ P\parallel Q \bk{ \tau[i]} P'\parallel Q' }
  \and
        \inferrule*[left=\scriptsize{(Res$^\bullet$)}]
  {P \bk{\mu[i]} P' \and \mu \not \in \{a,\co{a}\}}{ (\nu a)P \bk{ \mu[i]} (\nu a)P'} \and
  \inferrule*[left=\scriptsize{(Const$^\bullet$)}]
    {A \triangleq P \and P' \bk{\mu[i]} P}{ P' \bk{ \mu[i]} A}
\end{mathpar}
\caption{$\rlang$ backward semantics, symmetric rules for  \textsc{Par} and \textsc{Syn} are omitted.}
\label{fig:sembk}
\end{figure}

\begin{example}\label{ex:bk}
 Consider now the system $S$ in Example \ref{ex:fw}, which
 can retract its decision of consuming message $\tup{a}$
 and consume the message $\tup{b}$ as follows:
\begin{align*}
S \bk{\tau[w]}  \;&a.\nil + b.\nil    \parallel \co{a}\ckey{i}.(\co{b}\ckey{j}.\nil \parallel \ckey{j}\tup{b})\parallel \ckey{i}\tup{a}\\
\fw{\tau[z]}  \;&a.\nil + b\ckey{z}.\nil    \parallel \co{a}\ckey{i}.(\co{b}\ckey{j}.\nil \parallel \ckey{j}\tup{b}\ckey{z})\parallel \ckey{i}\tup{a}
\end{align*}
It is  worth noting that the message $\ckey{i}\tup{a}$ cannot be reversed in the reached processes, since the prefix $\co{a}\ckey{i}$ cannot be undone after having caused $\co{b}\ckey{j}$. Therefore, in order to revert the emission on $a$, it is first necessary to revert the emission on $b$. This enforces a cause-respecting reversal discipline.
\end{example}

Not all systems that can be written using the grammar in
Figure~\ref{fig:raccs-syn} are meaningful. For instance, $a.b\ckey{i}.\nil$ is
not. Therefore, we introduce the notion of \emph{reachable} systems, namely
those that are derivable from a standard one.

\begin{definition}[Reachable System]
  A system $R$ is said to be \emph{reachable} if it can be derived from a
  standard system using the rules reported in Figures~\ref{fig:semfw} and
  \ref{fig:sembk}.
\end{definition}

From now on, we only consider reachable systems.

\begin{lemma}\label{lem:keys}
Let $R$ and $S$ be two reachable systems. We have:
\begin{itemize}
	\item if $R\fw{\mu[i]} S$ then $\key(S) = \key(R) \cup \{i\}$
	\item if $R\bk{\mu[i]} S$ then $\key(S) = \key(R) \setminus \{i\}$
\end{itemize}
\end{lemma}

Every reachable system induces a partial order on its keys, which reflects
the causal execution order of the actions that produced it, as formally
stated in the definition below.

\begin{definition}[Ordering on keys]
The function $\po{\cdot} \colon \mathit{Systems} \to 2^{\mathcal{K} \times \mathcal{K}}$
is inductively defined as follows:
\begin{align*}
    &\po{\ckey{i}\tup{a}\ckey{j}} = \{(i,j)\}
	 \;\;\quad\quad\quad \po{A} = \emptyset &&  \;\;\quad\quad\quad \po{\ckey{i}\tup{a}} = \{(i,i)\} \\
	&\po{R} = \emptyset \;\text{ if } \std(R)  \;\;\quad\quad\quad  \po{(\nu a)R}= \po{R} 
	&&  \;\;\quad \po{R \parallel S} = \po{R} \cup \po{S}\\
	&\quad\quad\po{\mu\ckey{i}.P} = \{(i,j) \mid \,j\in \key(P)\} \cup \po{P} &&  
	\;\;\quad \po{\displaystyle{\sum}_{i\in I}\pi_i.P_i} = \bigcup_{i\in I} \po{P_i} 
\end{align*}
We let $\leq_{R}$ denote the reflexive and transitive closure of $\po{R}$.
\end{definition}

Note that the pair $(i, j)$, denoted by $i < j$, implies that the action with
key $i$ is a \emph{cause} of the action with key $j$. This accounts for both
the \emph{structural causality} inherent in the process (e.g., a prefix causes
its continuation) and the \emph{causality arising} from message consumption;
specifically, $\ckey{i}\tup{a}\ckey{j}$ indicates that the output action with
key $i$ precedes the input action with key $j$.

The following result establishes that $\leq_{R}$ is indeed a partial order on
the set of keys.

\begin{proposition}

  Let $R$ be a reachable system. Then $(\key(R),\leq_{R})$ is a partial order.
\end{proposition}
\begin{proof}
  Since  $\leq_{R}$ is defined as  the reflexive and transitive closure of $\po{R}$ it remains to show that it is antisymmetric. If $R$ is reachable, then there exists a standard process $S$ s.t. $S \red{\ }^* R$. The proof follows by induction on the length of the derivation.
  
  \textit{Base case}: $R = S$ and $\po{S} = \emptyset$.  The thesis follows trivially.

  \textit{Inductive step}. Then, $S \red{\ }^* R' \red{\ } R$ and
  $(\key(R'),\leq_{R'})$ is a partial order. The proof proceeds by case
  analysis on the transition $R' \red{} R$. Without loss of generality, we
  only consider forward rules, as these are the only ones that introduce new
  keys.
  If the last applied rule is (\textsc{Out}) then $R$ is of the form
  $a\ckey{i}.P\parallel \ckey{i}\tup{a}$ and we have that the new key $i$ is
  the element of the partial order
  $(\{i\},\{(i,i)\})$, 
  as $\ckey{i}\tup{a}$ does not induce
  any order, beside the pair $(i,i)$, and $P$ is standard.
  If the rule is (\textsc{msg}) then the order is clearly a partial order as
  $\po{\ckey{i}\tup{a}\ckey{j}} = \{(i,j)\}$ and the new key $j$ is set as the
  maximal element of the partial order.
  In the case of (\textsc{In}) and (\textsc{Tau}) the new key is the new
  minimal element.
  If the rule is (\textsc{Sum}) it is enough to observe that the key is added
  in one of the summand and the other are standard, hence they do not
  contribute to the partial order.
  If the rule is (\textsc{Par}) the thesis follows as the keys are disjoint.
  If the rule is (\textsc{Syn}) we have the partial orders
  $(\key(P'),\leq_{P'})$ and $(\key(Q'),\leq_{Q'})$, and the new key $i$ is
  the maximal in $(\key(Q'),\leq_{Q'})$ and the minimal in
  $(\key(P'),\leq_{P'})$. As $\key(P')\cap\key(Q') = \{i\}$ the thesis
  follows.
  If the rule was (\textsc{Res}) or (\textsc{Const}) the thesis follows
  easily. Finally assume that the rule is (\textsc{Act}), then the new key $i$ is added
  to $(\key(P'),\leq_{P'}$ which has the key $j$ as minimal element. As 
  $\po{\mu\ckey{i}.P'} = \{(i,j) \mid \,j\in \key(P')\} \cup \po{P'})$ we have the thesis.
\end{proof}

We conclude this section by defining the notion of  \emph{choice context}.

\begin{definition}[Choice Context]\label{def:ctx}
A \emph{choice context} is a process with a hole $\bullet$ generated by the following grammar:
\begin{align*}
  \ctx ::= & \mu\ckey{i}.\ctx \,\mid\, R \parallel \ctx  \,\mid\,  \nu a (\ctx) \,\mid\, \mu_j\ckey{k}.\ctx + \displaystyle{\sum}_{I\setminus\{j\}} \mu_i.R_i\,\mid\, \bullet
\end{align*}
We write $\ctx[R]$ to denote the process obtained by replacing the hole $\bullet$ in $\ctx$ with $R$.
\end{definition}

\section{$\rlang$ Properties}\label{sec:prp}

In this section, we investigate the properties of $\rlang$. We begin with a
fundamental property for any reversible calculus: every action can be undone.
Formally:

\begin{lemma}[Loop Lemma]\label{lm:loop}
Let $R$ and $S$ two  reachable systems. Then $R \fw{\mu[i]} S$ iff $S \bk{\mu[i]} R$. 
\end{lemma}
\begin{proof}
The proof proceeds by case analysis on the transition rule applied.
The analysis is straightforward, as each forward rule has a unique
corresponding backward rule, and vice versa.
\end{proof}

Another fundamental property is \emph{causal-consistent reversibility}~\cite{rccs}.
This property essentially states that the calculus stores the correct
amount of causal information; indeed, a notion of \emph{causal consistency}
is required to assess it. To discuss this property, we borrow several
definitions from~\cite{rccs}.
We use $t, t', s, s'$ to range over transitions. In a transition
$t : R \red{\mu[i]} S$ we call $R$ the \emph{source} of the transition, and
$S$ the \emph{target} of the transition. Two transitions are said to be
\emph{coinitial} if they have the same source, and {cofinal} if they have the
same target.
Given a transition $t$, we denote with $\underline{t}$ its inverse, i.e.,
if $t:R \fw{\mu[i]} S$ (resp. $t:R \bk{\mu[i]} S$) then
$\underline{t}:S \bk{\mu[i]} R$ (resp. $\underline{t}:S \fw{\mu[i]} R$).

We now define a notion of \emph{independence} between coinitial $\rlang$
transitions, based on a causality preorder on keys (inspired by~\cite{LaneseP21}).
Intuitively, independent transitions can be executed in any order
(formally stated in Proposition~\ref{prp:square}) whereas non-independent
transitions represent a \emph{choice}: the execution of one precludes the other.

\begin{definition}[Conflict]\label{def:confl}
Given an $\rlang$ reachable system $R$, two coinitial transitions
$t$ and  $s$ are \emph{conflicting} if one of the following holds:
\begin{enumerate}
\item $t:R\fw{\tau[i]} S_1$, $s:R\fw{\tau[j]} S_2$ and there exists a message $\ckey{z}\tup{a}$ in $R$ such that $\ckey{z}\tup{a}\ckey{i}$ belongs to $S_1$ and $\ckey{z}\tup{a}\ckey{j}$ belongs to $S_2$;

\item $t:R\fw{\mu[i]} S_1$, $s:R\fw{\eta[j]} S_2$, and the two transitions are generated by the same sum operator;
\item $t:R\fw{\mu[i]} S_1$, $s:R\bk{\eta[j]} S_2$ and $j\leq_{S_1} i$ or vice versa.
\end{enumerate}
\end{definition}

Let us comment on the notion of conflicting transitions of Definition
\ref{def:confl}. The first item tells us that two (forward) consumptions 
of the same message (marked by the key $z$) are in conflict.  The second item just says that all
the branches of a choice operator are in conflict with each other. For
example, if we take the process $R =a.\nil + b.\nil$, we have that
$R\fw{a[i]} a\ckey{i}.\nil + b.\nil$ and $R\fw{b[j]} a.\nil + b\ckey{j}.\nil$
and these two transitions are in conflict.
The last item tells us that two transitions are in conflict when a reverse step
eliminates some causes of a forward step. For example, the process
$a\ckey{i}.b.\nil$, can do a forward step with label $b\ckey{j}$, or a
backward step with label $a\ckey{i}$, and we have that $i\leq j$

\begin{definition}[Independence]\label{def:independence}
 Given a $\rlang$ reachable system $R$, two coinitial transitions
$t$ and  $s$ are \emph{independent}, written $t \ind s$, if they are not conflicting.
\end{definition}

\subsection{Causal Consistency}
Intuitively, coinitial and cofinal computations share the same causal information.
Hence, we require that they admit the same reversals. To formalize this,
we first introduce the notions of \emph{path} and \emph{path equivalence}.

We let $\chi, \omega$ range over sequences of transitions, which we call
\emph{paths}. We denote the empty sequence  by
$\epsilon$. 
We denote the number of transitions in a path $\chi$ by $|\chi|$. 
When $\chi$ is forward, i.e. composed by forward only transitions, we denote
by $\underline{\chi}$ the corresponding backward path, where the order of the transitions is reversed.
Moreover, we
let $\chi_1\chi_2$ denote the composition of two paths, $\chi_1$ and $\chi_2$,
provided they are \emph{composable}; that is, when the target of $\chi_1$
coincides with the source of $\chi_2$.

\begin{definition}[Causal Equivalence]
  Let $\causeq$ be the smallest equivalence on paths closed under composition
  and satisfying:
  \begin{description}
  \item[swap:] 
  $ts' \causeq st'$ for every two coinitial independent transitions  
  $t: R \red{\mu[i]} R_1$ and $s: R \red{\eta[j]} R_2$ 
  and every two cofinal transitions
   $s': R_1 \red{\eta[j]} S$, $t': R_2 \red{\mu[i]} S$;
  \item[cancellation:] $t\underline{t} \causeq \epsilon$ and
    $\underline{t}t \causeq \epsilon$
 \end{description}

\end{definition}

Intuitively, two paths are \emph{causally equivalent} if they differ only by
the reordering of independent transitions or the
addition (or removal) of \emph{do--undo} and \emph{undo--redo} transition pairs.

\begin{definition}[Causal Consistency (CC)]
An LTS is causal consistent if for any coinitial
and cofinal paths $\chi$ and $\omega$ we have $\chi \causeq \omega$.
\end{definition}

We can now prove causal consistency, using the theory in \cite{LanesePU24}.
Given an LTS in which loop lemma holds, and with a notion of independence on transitions,
thanks to theory in \cite{LanesePU24} one can derive  causal consistency  from  three basic properties:
\begin{itemize}
\item SP (Square Property),
\item BTI (Backward Transitions are Independent),
\item WF (Well Foundedness).
\end{itemize}

The Square Property tells that two coinitial independent transitions commute, thus
closing a diamond. Formally:

\begin{property}[Square Property - SP]\label{prp:square}
Given a $\rlang$ reachable system $R$ and two coinitial transitions
$t:R\red{\mu[i]} R_1$ and $s:R\red{\eta[j]} R_2$ with $t \ind s$, then there exist
two cofinal transitions $t':R_1 \red{\eta[j]} S$ and $s':R_2 \red{\mu[i]} S$ for some reachable system $S$.
\end{property}
\begin{proof}
 Assume $t:R\red{\mu[i]} R_1$ and $s:R\red{\eta[j]} R_2$ and $t \ind s$. We proceed by cases on
 the kind of the applied transitions: either forward or backward.
 \begin{enumerate}
  \item $t:R\red{\mu[i]} R_1$ and $s:R\red{\eta[j]} R_2$ are both forward. As $t \ind s$
  we know that the two transitions do not imply consuming the same message
  nor they are different branches of the same sum. We proceed by case analysis on the forward rules. We just show the most significative cases.

  If the two transition are an emission of a message, w.l.o.g, we can assume that
  $R = (\nu \tilde{a})(R_l \parallel R_r)$ with $\tilde{a} $ a set of restricted names, and
  $R_l\fw{\mu[i]} R'_l$ and $R_r\fw{\eta[j]} R'_r$. It is easy to see that by applying twice
  the rule (\textsc{Par}) we have:
  \begin{align*}
&  \nu \tilde{a}(R_l \parallel R_r)\fw{\mu[i]}\nu \tilde{a}(R'_l \parallel R_r)\fw{\eta[j]}\nu \tilde{a}(R'_l \parallel R'_r) \quad \text{ and }\\
  &  \nu \tilde{a}(R_l \parallel R_r)\fw{\eta[j]}\nu \tilde{a}(R_l \parallel R'_r)\fw{\mu[i]}\nu \tilde{a}(R'_l \parallel R'_r)
  \end{align*}

  If the two reductions consume  different messages we have  to distinguish two cases:
  either the two communications happen in two separate parts of the system, or not. In the first case
  we can proceed as in the previous case by assuming $R =  \nu \tilde{a}(R_l \parallel R_r)$ with the two parallel systems not interacting each other. In the second case instead, we have that the two systems interact, that is the message consumed by $R_r$ is present in $R_l$ and vice versa. In this case
  we have that $R = \nu \tilde{a}(\mathtt{C}_1[R_l] \parallel \mathtt{C}_2[R_r])$ with
  $\mathtt{C}_1[\bullet]$ and $\mathtt{C}_2[\bullet]$ being two choice contexts (see Definition \ref{def:ctx}). Also in this case we have that both
  $R_l$  and $R_r$ can do two different and independent forward transitions. We can then conclude by making the cases analysis on smaller terms.

  \item $t:R\red{\mu[i]} R_1$ is forward and $s:R\red{\eta[j]} R_2$ is backward.
  By case analysis on the applied rules. The reasoning is similar to the forward case.

  \item $t:R\red{\mu[i]} R_1$ is backward and $s:R\red{\eta[j]} R_2$ is forward.
    By case analysis on the applied rules. The reasoning is similar to the previous case.

  \item $t:R\red{\mu[i]} R_1$ and $s:R\red{\eta[j]} R_2$ are both backward.  By case analysis on the applied rules. The reasoning is similar to the previous case.
 \end{enumerate}
\end{proof}

BTI generalises the concept of backward determinism used for reversible
sequential languages. It specifies that two backward transitions from a same
configuration are always independent.

\begin{property}\label{prp:BTI}[Backward Transitions are Independent - BTI]
  Given a reachable $\rlang$ system $R$, any two distinct coinitial backward
  transitions $t : R \bk{\mu[i]} S_1$ and $t : R \bk{\eta[j]} S_2$ are
  independent.
\end{property}

The \emph{BTI property} holds trivially: according to the definitions of
conflicting and independent transitions (Definitions~\ref{def:confl} and~\ref{def:independence}),
there are no cases in which two backward transitions are in conflict.
Consequently, any pair of backward transitions is always independent.

We now show that reachable configurations have a finite past.

\begin{proposition}[Well-Foundedness - WF]\label{prp:wf}
  Let $R_k$ be a reachable $\rlang$ system. There is no infinite sequence
  starting from $R_k$ such that $R_{i} \bk{\mu_i[j_i]} R_{i-1}$ for all
  $i \ge k$.
\end{proposition}

\begin{proof} Well-foundedness (WF) follows from the fact that each backward transition removes
a key. 
Thanks to Lemma~\ref{lem:keys}
we have that number of keys $|\key(R)|$ in a reachable system $R$ is finite,
only a finite number of backward steps can be performed.
\end{proof}

The \emph{Parabolic Lemma}~\cite[Lemma 11]{rccs} states that any path is causally
equivalent to a  path consisting of a sequence of backward steps
followed by a sequence of forward steps. In other words, up to causal equivalence,
paths can be rearranged to first reach the maximum freedom of choice by moving
backwards, before proceeding forward. 
Here we state the following.

\begin{property}[Parabolic Property - PP] An LTS satisfies the Parabolic Property
  iff for any path $\chi$, there exist two forward-only paths $\omega,\omega'$
  such that $\chi \causeq \underline{\omega}\omega'$ and
  $|\omega| + |\omega'| \leq |\chi|$. 
\end{property}
Observe that if the LTS of  reversing calculi has the Parabolic Property then the parabolic lemma holds. 

We can now prove that the parabolic lemma holds, thanks to the proof schema of
\cite{LanesePU24} recalled in the following propositions.

\begin{proposition}[cf. Proposition 3.4 \cite{LanesePU24}] Suppose BTI and SP
  hold, then PP holds.
\end{proposition}

\begin{proposition}[cf. Proposition 3.8 \cite{LanesePU24}] Suppose WF and PP
  hold, then CC holds.
\end{proposition}

Summarizing our achievements we have:
\begin{theorem}\label{th:PL}
  The LTS of Definition~\ref{de:ltlrccsa} satisfies PP.
\end{theorem}
\begin{proof}
 This is a consequence of properties \ref{prp:BTI} and \ref{prp:square}.
\end{proof}

\begin{theorem}\label{th:cc}
  The LTS of Definition~\ref{de:ltlrccsa} satisfies CC.
\end{theorem}
\begin{proof}
 This is a consequence of theorem~\ref{th:PL} and property \ref{prp:wf}.
\end{proof}

\added{Theorem \ref{th:cc} shows that causal equivalence characterizes a space for admissible rollbacks
that are (i) correct as they do not lead to states not reachable by some forward computation
and (ii) flexible enough to allow undo operations to be rearranged with respect to the order
in which the undone concurrent transitions were originally performed. This implies that
the states reached by any backward computation could be reached by performing forward
computations only (see the following Corollary). Therefore, we can conclude that $\rlang$ meets causal reversibility.}


\begin{corollary}
A system $S$ is reachable iff there exists a standard system $R$ and a forward only path $R \fw{}^*S$.
\end{corollary}
\begin{proof} From PP, by noticing that the backward path is empty since $R$ cannot take backward
actions.
\end{proof}

\begin{example}
Let us consider the system   $R = (a.\nil + b.\nil) \parallel \co{a}.\co{b}.\nil $, of Example \ref{ex:intro} and the following two paths:
\begin{align*}
&\chi = R\fw{\tau[i]}
 \fw{\tau[j]}
 \fw{\tau[z]}  (a.\nil + b\ckey{z}.\nil) \parallel \co{a}\ckey{i}.(\co{b}\ckey{j}.\nil \parallel \ckey{j}\tup{b}\ckey{z})\parallel \ckey{i}\tup{a} \\
 &\omega = R\fw{\tau[i]}
 \fw{\tau[j]}
 \fw{\tau[w]}  (a\ckey{w}.\nil + b.\nil) \parallel \co{a}\ckey{i}.(\co{b}\ckey{j}.\nil \parallel \ckey{j}\tup{b})\parallel\ckey{i}\tup{a}\ckey{w} \\
 &\qquad  \bk{\tau[w]}\fw{\tau[z]}   (a.\nil + b\ckey{z}.\nil) \parallel \co{a}\ckey{i}.(\co{b}\ckey{j}.\nil \parallel \ckey{j}\tup{b}\ckey{z})\parallel \ckey{i}\tup{a}
\end{align*}
Since $\chi$ and $\omega$ are coinitial, they start in $R$, and cofinal, they end up in the same system,  thanks to Theorem~\ref{th:cc} we have that they are causally equivalent, that is $\chi \causeq \omega$.
\end{example}

\section{$\rlang$ at work: a leader election example}\label{sec:ex}
We illustrate the behaviour of $\rlang$ with a simple race protocol modelling
a leader election scenario, by adapting the example of \cite{broad}
\footnote{\added{In \cite{broad} the modelling is simpler, as the considered reversible calculus uses broadcast, e.g n-ary syncrhonissation, as communication facility.}}.
 Several
processes compete to react to a message announcing an election. The process
that reacts first proposes itself as the leader (by emitting $n-1$ messages
$\co{p}_i$), while the others processes may choose to follow its lead (by
replying $\co{ok}_i$) or to refuse it (by replying $\co{no}_i$). This
 may naturally lead to a deadlock, but thanks to the reversible semantics the
system can get back from a global deadlock and try other solutions.
Also, when a process proposes itself as a leader (via the prefix $\co{p}_i$),
it must collect $n-1$ confirmative replies before considering itself
the new leader.
Let $\mathtt{seq}(a,n) = a^1.a^2. \dots .a^n.\nil$ denote the sequential
composition of $n$ instances of the prefix $a$.

\begin{eqnarray} \label{lbl:leader} R_i \triangleq \big( l.(
  \displaystyle{\prod}_{j\in I \setminus\{i\}} \co{p_i}.ok_j.\co{a}_i)
  \parallel \texttt{seq}(a_i, | I | -1 )\big) + \displaystyle{\sum}_{j\in I
    \setminus\{i\}} p_j.(\co{ok}_i + \co{no}_i)
\end{eqnarray}
As an example, let us suppose a system made of 3 processes, that is:
$$\mathtt{S}= R_1 \parallel R_2 \parallel R_3 \parallel \co{l}.\nil$$
where $\co{l}$ acts as the token emitted to initiate an election. Each $R_i$
can become a leader by consuming the message $\tup{l}$ (left branch of the
main choice of (\ref{lbl:leader})) or react to someone else's proposal
(right-branch of (\ref{lbl:leader})) ). We now describe an execution in which
$R_1$ proposes itself as leader, $R_2$ agrees to it and $R_3$ refuses. This
leads the entire system to a bad state $\mathtt{S}_{e}$ where: $R_1$ has not
enough confirmation to become leader, no other process can take over and the
entire system is deadlocked.
In what follows we will use 
use natural numbers as keys, just to ease the reading of the example.
\begin{align*}
\mathtt{S}_{e}=\; & l\ckey{2}.\big(( \co{p}_1\ckey{3}.ok_2.\co{a_1} \parallel p_1\ckey{4}.ok_3.\co{a_1} \parallel \ckey{3}\tup{p_1}\ckey{5} \parallel \ckey{4}\tup{p_1}\ckey{6}) \parallel \mathtt{seq}(a_1,2)\big) +
	\displaystyle{\sum}_{I \setminus\{1\}} p_j.(\co{ok}_i + \co{no}_i)\\
	&\parallel \big( l.( \displaystyle{\prod}_{j\in I \setminus\{2\}} \co{p_i}.ok_j.\co{a}_i) \parallel \texttt{seq}(a_2, 2 )\big) +
	p_1\ckey{3}.(\co{ok}_2\ckey{8} + \co{no}_2 \parallel \ckey{8}\tup{ok_2})
	+
	p_3.(\co{ok}_2 + \co{no}_2)
	\\
	&\parallel  \big( l.( \displaystyle{\prod}_{j\in I \setminus\{3\}} \co{p_i}.ok_j.\co{a}_i) \parallel \texttt{seq}(a_3,2 )\big) +
	p_1\ckey{4}.(\co{ok}_3 + \co{no}_3\ckey{7} \parallel \ckey{7}\tup{no_3})
	+
	p_2.(\co{ok}_3 + \co{no}_3)
	\\
	&\parallel \co{l}\ckey{1}.\nil \parallel \ckey{1}\tup{l}\ckey{2}
\end{align*}
Thanks to the reversible semantics, process $R_3$ can revert its decision and agree with $R_1$ and let the system reach a valid state as follows:
$$ \mathtt{S}_{e} \bk{\tau\ckey{7}}\fw{\tau\ckey{9}}\mathtt{S}_{ok}$$
\begin{align*}
\mathtt{S}_{ok}=\; & l\ckey{2}.\big(( \co{p}_1\ckey{3}.ok_2.\co{a_1} \parallel p_1\ckey{4}.ok_3.\co{a_1} \parallel \ckey{3}\tup{p_1}\ckey{5} \parallel \ckey{4}\tup{p_1}\ckey{6}) \parallel \mathtt{seq}(a_1,2)\big) +
	\displaystyle{\sum}_{I \setminus\{1\}} p_j.(\co{ok}_i + \co{no}_i)\\
	&\parallel \big( l.( \displaystyle{\prod}_{j\in I \setminus\{2\}} \co{p_i}.ok_j.\co{a}_i) \parallel \texttt{seq}(a_2, 2 )\big) +
	p_1\ckey{3}.(\co{ok}_2\ckey{8} + \co{no}_2 \parallel \ckey{8}\tup{ok_2})
	+
	p_3.(\co{ok}_2 + \co{no}_2)
	\\
	&\parallel  \big( l.( \displaystyle{\prod}_{j\in I \setminus\{3\}} \co{p_j}.ok_j.\co{a}_i) \parallel \texttt{seq}(a_3,2 )\big) +
	p_1\ckey{4}.(\co{ok}_3\ckey{9} + \co{no}_3 \parallel \ckey{9}\tup{ok_3})
	+
	p_2.(\co{ok}_3 + \co{no}_3)
	\\
	&\parallel \co{l}\ckey{1}.\nil \parallel \ckey{1}\tup{l}\ckey{2}
\end{align*}
where the two messages $\tup{ok_2}$ and $\tup{ok_3}$ are emitted respectively by $R_2$ and $R_3$, and can be consumed by $R_1$.

\section{Conclusions}\label{sec:conc}

In this paper, we investigate the integration of reversibility with asynchrony
in CCS. As a design choice, we model asynchrony at the semantic level, making
CCSa more closely aligned with real-world programming languages. This approach
is consistent with recent work on session types, where messages in
transit--i.e., messages that have been sent but not yet received--are
explicitly considered in the setting of asynchronous session types~\cite{ChenDSY17,PY26}.

Most existing reversible calculi either assume synchronous communication~\cite{rccs,ccsk,CristescuKV13,GalindoNST21}, or model asynchronous
communication in a restricted way, for instance by disallowing continuations
after output prefixes~\cite{LaneseMS10,CardelliL11}. In this paper, we take a
first step toward defining a reversible semantics for asynchronous variants of
CCS in which message emission and reception are decoupled events. This
separation introduces additional causal dependencies that must be carefully
accounted for in the reversible setting. Existing approaches to derive
reversible semantics for CCS, such as~\cite{rccs,ccsk}, are not directly
applicable and fail to yield a reversible semantics for $\lang$. The static
approach of~\cite{ccsk} fails since the emitting rule does not fit with the
considered SOS format. This is due to the fact the arguments of the rule
change and in the conclusion we have a additional argument, that is the emitted
message. The dynamic approach of~\cite{rccs} fails since the memory mechanism
used to log events is devised for synchronisations, and cannot deal with
partial events such as a message emission. Hence we have devised an ad-hoc
reversible semantics by leveraging the key ideas of~\cite{ccsk}: making static
all the operators of the calculus and using identifiers to mark events. This
leads to the use of half marked messages of the form $\ckey{i}\tup{a}$,
indicating a message in transit, and fully marked messages, of the form
$\ckey{i}\tup{a}\ckey{j}$, indicating a message which has been consumed. We
remark that in the framework of~\cite{ccsk} elements tagged with two keys,
like consumed messages $\ckey{i}\tup{a}\ckey{j}$,  do not exist and contrast
the assumptions of the approach. We have then showed that the obtained
reversible semantics is causally consistent.

Several directions for future work remain open. One possible extension is to
study behavioural equivalences for $\rlang$, for instance by adapting
forward-reverse bisimulation notion~\cite{ccsk} to the asynchronous setting in
the flavour of~\cite{api}. Another direction is to investigate a truly
concurrent semantics for $\lang$ and $\rlang$, along the lines of recent work
on RCCS~\cite{MelgrattiMP24}, \added{or by reinterpreting the encodings of
\cite{BaldanBGM15} with the notion of asynchrony used in this paper in the setting of
reversible Petri nets \cite{MelgrattiMP26}.}

Also, another promising direction is to
investigate the impact of reversibility in variants of CCS equipped with
Linda-like coordination primitives~\cite{BusiGZ98,BusiZ08}, inspired by the
Linda model of shared tuple spaces and compare it with \cite{GiachinoLMT17}. On the practical side, it would be
interesting to explore connections with programming languages that rely on
asynchronous message passing, such as Erlang or Go, in order to better
understand how reversible semantics could support debugging or rollback
mechanisms in distributed systems. This would require to add channel queues,
in the case of Go, or mailboxes, in the case of Erlang, to $\lang$.

\nocite{*}
\bibliographystyle{eptcs}
\bibliography{generic}
\end{document}